\documentclass[
  a4paper
]{amsart}

\usepackage[utf8]{inputenc}

\usepackage{amsmath}
\usepackage{amsfonts}
\usepackage{amssymb}
\usepackage{amsthm}
\usepackage{mathabx}
\usepackage{mathrsfs}

\usepackage{graphicx}
\usepackage{subcaption}
\usepackage{cite}

\usepackage{enumerate}

\newcommand{\setint}[1]{
\mathchoice{\ring{#1}}
           {{#1}^\circ}
           {{#1}^\circ}
           {{#1}^\circ}%
}
\newcommand{\contdiff}[1]{%
\if\relax\detokenize{#1}\relax
\ensuremath{\mathcal{C}}%
\else
\ensuremath{\mathcal{C}^{#1}}%
\fi%
}
\def\nashEq/{\ensuremath{\mathcal{E}}}
\newcommand{\RR}{\mathbb{R}}
\newcommand{\NN}{\mathbb{N}}
\renewcommand\epsilon{\varepsilon}
\newcommand{\fixedP}[1]{%
\ensuremath{\mathscr{F}(#1)}%
}

\newtheorem{prop}{Proposition}[section]
\newtheorem{thm}[prop]{Theorem}
\newtheorem{lemma}[prop]{Lemma}
\newtheorem{coroll}[prop]{Corollary}
\theoremstyle{definition}
\newtheorem{defn}[prop]{Definition}
\newtheorem{remark}[prop]{Remark}
\newtheorem*{remark*}{Remark}
\newtheorem*{exmpl*}{Example}
\newtheorem{exmpl}[prop]{Example}

\def\propMut/{{\ensuremath{(A')}}}

\usepackage[english]{babel}
\usepackage{csquotes}

\numberwithin{equation}{section}

\begin{document}

\pagestyle{headings}

\title[Mutation Limits in Evolutionary Games]{The Stabilisation of 
Equilibria in Evolutionary Game 
Dynamics through Mutation: Mutation Limits in Evolutionary Games
}
\author{Johann Bauer}
\address{Johann Bauer, Department of Mathematics, City, University of 
London}
\email{johann.bauer@city.ac.uk}
\author{Mark Broom}
\address{Mark Broom, Department of Mathematics, City, University of 
London}
\author{Eduardo Alonso}
\address{Eduardo Alonso, Department of Computer Science, City, 
University of London}

\thanks{\emph{Acknowledgements.}
%
MB and JB would like to thank their hosts at the Department of 
Mathematics and Statistics at the University of North Carolina at 
Greensboro, where parts of the research were undertaken,
and acknowledge the support by the European Union’s 
Horizon 2020 research and innovation programme under the Marie 
Skłodowska-Curie grant agreement No 690817, as part of the Research and 
Innovation Staff Exchange (RISE) programme.
}
%



\begin{abstract}
The multi-population replicator dynamics (RD) can be considered a
dynamic approach to the study of multi-player games, where it was shown 
to be related to Cross' learning, as well as of systems of coevolving
populations.
However, not all of its equilibria are Nash equilibria (NE)
of the underlying game, and neither convergence to an NE nor 
convergence in general are guaranteed.
Although interior equilibria are guaranteed to be NE, no interior 
equilibrium can be asymptotically stable in the multi-population RD, 
resulting, e.g., in cyclic orbits around a single interior NE.
We introduce a new notion of equilibria of RD, called mutation limits, 
which is based on the inclusion of a naturally arising, simple form of 
mutation, but is invariant under the specific choice of mutation 
parameters.
We prove the existence of such mutation limits for a large range of 
games, and consider a subclass of particular interest, that of 
attracting mutation limits.
Attracting mutation limits are approximated by asymptotically stable 
equilibria of the (mutation-)perturbed RD, and hence, offer an 
approximate dynamic solution of the underlying game, especially if the 
original dynamic has no asymptotically stable equilibria.
In this sense, mutation stabilises the system in certain cases and 
makes attracting mutation limits near-attainable.
Furthermore, the relevance of attracting mutation limits as a game 
theoretic equilibrium concept is emphasised by a similarity of 
(mutation-)perturbed RD to the Q-learning algorithm in the context of
multi-agent reinforcement learning.
In contrast to the guaranteed existence of mutation limits, 
attracting mutation limits do not exist in all games, raising the 
question of their characterization. 
\end{abstract}

\maketitle

\section{Introduction}

Evolutionary game theory has contributed significantly to our
understanding of a wide range of biological, e.g.,
\cite{maynard_smith_logic_1973, 
broom_game-theoretical_2013},
and social phenomena, as shown by the vast research 
into the evolution of cooperation and eusociality, e.g., 
\cite{axelrod_evolution_1981}, or the problem of collective action, e.g.,
\cite{pacheco_evolutionary_2009}.
The evolutionary game theoretic approach,
formulated in \cite{maynard_smith_logic_1973},
initially assumed a single population with intrapopulation interaction and 
competition for reproduction, resulting in the concept of the 
evolutionarily stable strategy (ESS), a refinement of the Nash 
equilibrium concept, where a 
strategy is said to be evolutionarily stable if it outperforms any other 
newcomer strategy in a population consisting almost entirely of players playing 
the former.
While the intuition underlying the notion of an ESS is dynamic,
its main definition is usually given in static terms.
In an effort to capture the dynamic intuition of the ESS concept, the 
continuous time replicator dynamics (RD), provided by 
\cite{taylor_evolutionary_1978}, relates the ESS to certain stationary 
points, 
\cite{hofbauer_evolutionary_1998}, albeit lacking a complete characterization.
In its usual formulation, it captures the single population setting with 
pairwise intrapopulation interactions.
However, just as the concept of an ESS has been extended to the 
multi-population, or multi-species, setting, e.g., 
\cite{cressman_stability_1992}, so has RD been formulated and analysed in the 
multi-population setting with intrapopulation competition (for reproduction) 
but \emph{inter}population interactions (determining reproductive advantage), 
e.g., \cite{weibull_evolutionary_1995}.
Forms of multi-population RD have been employed in the analysis of 
coevolutionary systems,
such as mutualism \cite{bergstrom_red_2003},
antagonistic coevolution of host-parasite systems 
\cite{nee_antagonistic_1989, song_host-parasite_2015},
of institutional ecosystems \cite{hashimoto_theoretical_2017},
of the evolution of a population's sex ratio \cite{argasinski_dynamics_2012},
or the coevolution of social behaviour and recognition 
\cite{smead_coevolution_2016}.
It has further been linked to Cross' learning, a simple type of 
reinforcement learning \cite{borgers_learning_1997}.

In the context of potentially very large systems, e.g., complex 
ecosystems or multi-agent systems, 
multi-population RD is of special interest because a population's 
composition 
evolves exclusively depending on the payoffs from interactions, but independent 
of any information about the other populations' payoffs, their compositions, or 
indeed their very existence.
The latter specifics affect a population's composition only through their 
effect on its payoffs.
Borrowing the term from \cite{hart_uncoupled_2003}, we call this property of RD
its \emph{uncoupledness}.

In spite of RD leading to payoff-improving or even equilibrium states in 
certain cases, there are intuitively simple games, for which neither an ESS 
exists nor RD reaches any Nash equilibrium, exhibiting periodic 
limit or general non-convergent behaviour instead:
In the usual rock-paper-scissors (RPS) game, RD has exclusively 
periodic orbits 
in the single population case and the only Nash equilibrium, an interior point, 
is not approached from any initial state, e.g., 
\cite{broom_game-theoretical_2013},
and a range of (un)-stable situations can result
\cite{hoffman_experimental_2015}.
Further, the two population setting results in periodic orbits, as well, 
and therefore does not reach the interior Nash equilibrium either.
An analogue result holds for the matching pennies game, e.g., 
\cite{weibull_evolutionary_1995}.
Indeed, it has been shown in \cite{hart_uncoupled_2003} that no 
uncoupled dynamics, in particular RD, can be converging to a Nash 
equilibrium for all possible games.
For our understanding of actual biological populations, this 
periodicity is not necessarily problematic.
On the contrary, periodic population dynamics similar to the 
single-population RPS case have been observed in nature, e.g., in the 
common side-blotched lizard (\emph{Uta stansburiana}) 
\cite{sinervo_rockpaperscissors_1996}. 
For our understanding of the conditions of behavioural convergence in 
multi-agent systems and their ability to solve large-scale problems 
such periodic behaviour is less desirable.

Although RD is intended to capture the idea of evolutionary selection,
and thus is inspired by evolution,
it treats mutation, an arguably central process of evolution and one of 
the main generators of the diversity on which selection operates, as an 
extremely rare event, to the degree that it is actually absent from the 
formulation of the dynamics, especially in the case of multiple 
populations, e.g., \cite{weibull_evolutionary_1995}.
Approaches which include mutation mainly focus on the single population case
\cite{hofbauer_selection_1985,%
burger_mutation-selection_1989,%
bomze_stability_1995,%
page_unifying_2002,%
boylan_evolutionary_1994,%
izquierdo_strictly_2011,%
allen_mutation_2012
},
consider a payoff-adjusted RD,
or a discrete time process \cite{burger_mathematical_1998},
or a single discrete population
\cite{imhof_evolutionary_2005, 
veller_finite-population_2016},
while we are not aware of an analysis of continuous-time 
multi-population RD with mutation,
apart from \cite{ritzberger_theory_1994} where certain approximations 
to multi-population RD are considered, with a different focus however 
and not linked to mutation.

We demonstrate that accounting for mutation in multi-population 
RD can fundamentally change the properties of the dynamics, 
i.e., preclude any periodicity in certain cases and, 
furthermore, guarantee convergence to states close to Nash 
equilibria, which would not be reachable under standard RD.
Note that the non-existence result in 
\cite{hart_uncoupled_2003} does not directly apply to such 
mutation dynamics, as it only considers Nash-convergence.

Our main interest, therefore, lies with the derivation of an uncoupled 
dynamics, which, on the one hand, explicitly considers mutation and, on the 
other hand, is as close as possible to standard RD, and with the analysis of 
how this mutation mechanism affects the position and stability of equilibria 
compared to the standard (multi-population) RD.
The resulting mutation mechanism with spontaneous mutations from one type to 
another is of course not appropriate for all biological mutation processes.
In a biological population, such spontaneous mutation between a finite number 
of types occurs, e.g., for single nucleotide polymorphisms, 
where alleles 
differ by only one nucleotide, with the number of possible 
single nucleotide 
polymorphisms at that position restricted to four.
Furthermore, such point mutations are known to 
occur with a non-negligible probability 
\cite{collins_rates_1994, crow_high_1997}
and can be significant factors in diseases, 
\cite{neel_mutation_1978, crow_high_1997}, e.g.,
sickle cell anaemia, \cite{marotta_human_1977, 
conner_detection_1983}, which 
also interacts with malaria parasites, \cite{luzzatto_sickle_2012},
cystic fibrosis, \cite{hamosh_cystic_1992},
or $\beta$ thalassemia, \cite{chang_beta_1979, 
sidore_genome_2015},
and further in human cancer cells, 
\cite{davies_mutations_2002, minde_messing_2011}.
There is further evidence that in \emph{Drosophila} most such 
nonsynonymous point mutations are deleterious, while the rest are 
slightly deleterious, near-neutral, or weakly beneficial, 
\cite{sawyer_prevalence_2007}, suggesting that a weak selection 
assumption as we employ can be reasonable for persisting polymorphisms.
Considered as a learning dynamics, modifications of multi-population RD 
have been shown to be linked to so-called Q-learning, a more 
sophisticated reinforcement learning algorithm,
\cite{tuyls_selection-mutation_2003}.
In particular, the resulting modification can be interpreted as a 
mutation-like term.

The inclusion of mutation should not only further our understanding of  
coevolutionary multi-population systems, such as ecosystems.
Its ability to stabilise equilibria for any non-zero mutation rate, and 
thereby make them attainable by an uncoupled dynamics, should also be 
useful in the study of game theoretical solution concepts, such as 
$\epsilon$-Nash equilibria, \cite{fudenberg_limit_1986}, and the 
formulation of conditions for the convergence of learning in 
multi-agent systems.
\\

We proceed by introducing the standard multi-population RD, i.e., without 
mutation, and recounting some stability properties of its equilibria and their 
relation to game theoretic concepts, such as Nash equilibria and evolutionary 
stability.

We then introduce mutation and give a heuristic derivation of the specific form 
of mutation we consider, defining a replicator-mutator dynamics (RMD), the 
equilibria of which we call \emph{mutation equilibria}. For fixed mutation 
parameters, we prove the existence of equilibria of RMD, their $\epsilon$-Nash 
property, and their uniqueness and asymptotic stability under very high 
mutation.

We proceed by defining the concept of limits of mutation equilibria for 
vanishing mutation, which we call \emph{mutation limits}.
Mutation limits and their properties are independent of any choice of specific 
mutation parameters.
We prove the existence of mutation limits for all systems with 
continuously differentiable fitness functions and give a sufficient 
condition for a Nash 
equilibrium to be a mutation limit.

In order to address the question of reachability of mutation limits, we define 
the notion of an \emph{attracting} mutation limit based on the asymptotic 
stability of the mutation equilibria by which it is approximated.
Such attracting mutation limits are reachable in the sense that for any choice 
of mutation parameters there is an asymptotically stable mutation equilibrium 
arbitrarily close to the mutation limit.

We further provide a sufficient condition for a Nash equilibrium to be 
an attracting 
mutation limit. In particular, all evolutionarily stable states are attracting 
mutation limits, but not all attracting mutation limits are evolutionarily 
stable, showing the notion to be a strictly weaker property than evolutionary 
stability.
We conclude by giving a necessary condition for attracting mutation limits, 
ruling out hyperbolic interior equilibria.

\section{Multi-population Replicator Dynamics}

In the following we consider the situation where we have a finite set 
of populations $I = \{ 1, 2, \ldots, N \}$ and each population $i$ 
consists of a finite number of types which we 
enumerate and denote by $S_i = \{ 1, 2, \ldots, n_i \}$.
Note that types are population-specific and numbers do not identify 
types across populations.
The composition of a population $i$ is then given as a vector $x_i$ 
such that $x_{ih} \geq 0$ gives the frequency of a type $h \in S_i$ in 
population $i$.
Thus, the set of possible compositions of population $i$ is given as:
\[ \Delta_i = \left\{ x_i \in \RR_{\geq 0}^{n_i} \Bigg| 
\sum_{h \leq n_i} x_{ih} = 1  \right\} \]
For convenience, we denote the Cartesian product of the $\Delta_i$ 
($i=1 \ldots N$) by 
$\Delta$, i.e., $\Delta = \bigtimes_{i \leq N} \Delta_i $, and 
denote by $\setint{\Delta}$ the interior of $\Delta$, i.e.,
$\forall i \leq N, h \leq n_i: x_{ih} > 0$.
Furthermore, we set $S = \{ (i,h) | i \in I \text{ and } h \in S_i \}$, 
such that $\Delta \subset \RR^S$.
The state of the multi-population model then is a description of the 
frequencies of the different types in the populations, i.e., it is 
given by some $x \in \Delta$.
\\

We assume that for each population $i \in I$ and each type in that 
population $h \in S_i$ we have a function
$f_{ih} \in \contdiff{1}( U , \RR )$, for $U \supset 
\Delta$ open, 
describing the reproductive rate or fitness $f_{ih}(x)$ of that type in 
a given state $x \in \Delta$ and we define population $i$'s average 
fitness as $\bar{f}_i(x) = \sum_{h \leq n_i} x_{ih} f_{ih}(x)$.
It should be noted that fitness is frequency-dependent in replicator 
dynamics models and not affected by population sizes.
We further assume that there is no intraspecific interaction affecting 
fitness in a type-specific manner, i.e., the fitness values of types in 
population $i$ are independent of the composition of population $i$ or 
$ \frac{\partial}{\partial x_{ik}} f_{ih} (x) = 0 $
($i \in I, h,k \in S_i$) 
in keeping with the classic normal-form game settings.\footnote{%
Note that this assumption is not essential for all results.%
}
The standard multi-population replicator dynamics, based on 
\cite{taylor_evolutionarily_1979} and developed later, e.g., 
\cite{weibull_evolutionary_1995}, is given by the 
following system of differential equations:
\begin{gather}\label{eqn:RD}\tag{RD}
\dot{x}_{ih} = \phi_{ih}(x) := x_{ih} \left( f_{ih}(x) - \bar{f}_i(x) 
\right)
\quad
( i \in I, h \in S_i )
\end{gather}
We denote by $\Phi: \RR \times \Delta \rightarrow \Delta$ the flow 
of \eqref{eqn:RD}, i.e., for $x \in \Delta$, $\Phi(\cdot, x): \RR 
\rightarrow \Delta, t \mapsto \Phi(t,x)$ is a 
solution of \eqref{eqn:RD} with $\Phi(0,x) = x$.
Due to our continuity assumption on $f$, the existence and uniqueness 
of $\Phi$ is clear, e.g., \cite[Thm 6.1]{teschl_ordinary_2012}.

\subsection{Stationary points of the replicator dynamics}

We  give a short recount of some well-known properties of
\eqref{eqn:RD} with regards to game theory, beginning with the main 
concept of game theory:

\begin{defn}[Nash equilibrium]
We call a state $x^* \in \Delta$ a \emph{Nash equilibrium} if
\[
\forall i \in I, z_i \in \Delta_i\backslash\{x_i^*\} :
\bar{f}_i(x^*) \geq \bar{f}_i( x_{-i}^*, z_i),
\]
where $(x_{-i}^*, z_i)$ denotes the state such that
\[
[x_{-i}^*, z_i]_{jk} = \begin{cases}
z_{ik}	&	\text{ if } j=i, \\
x^*_{jk}	&	\text{ otherwise }
\end{cases}.
\]
We call $x^* \in \Delta$ a \emph{strict Nash equilibrium} if all 
inequalities in the Nash equilibrium condition are strict.
\end{defn}

\begin{remark*}
It is clear that $x^* \in \Delta$ is a Nash equilibrium if and only if
\[
\forall i \in I, h \leq n_i :\: g_{ih}(x^*) := f_{ih}(x^*) - 
\bar{f}_i(x^*) \leq 0.
\]
Note that $g_{ih}(x)$ is exactly the coefficient of $x_{ih}$ in 
\eqref{eqn:RD}.
Therefore, we can denote the set of Nash equilibria by $\nashEq/ = 
\{x \in \Delta \,|\, g(x) \leq 0 \}$, where the inequality is 
component-wise.
A strict Nash equilibrium $x^* \in \Delta$ in particular is a state 
where each population consists of exactly one type, i.e., for each 
population $i \in I$ there is exactly one type $h_i$ such that 
$x^*_{ih_i} = 1$.
\end{remark*}

The following relationships between Nash equilibria and stationary 
points of \eqref{eqn:RD} are well-known:
\begin{prop}
If $x \in \Delta$ is a Nash equilibrium, then $x$ is a stationary point 
of \eqref{eqn:RD}, i.e., $\phi(x) = 0$.
\end{prop}

The other direction of this implication holds for interior stationary 
points, e.g. \cite[Thm 5.2]{weibull_evolutionary_1995}:
\begin{prop}
If $x \in \setint{\Delta}$ is a stationary point of \eqref{eqn:RD}, 
then $x$ is a Nash equilibrium.
\end{prop}

\subsubsection{Stability properties of equilibria}

Our special interest lies with the attainability of Nash equilibria.
Therefore, we restate a few stability properties of Nash equilibria and 
stationary points of \eqref{eqn:RD} respectively.

\begin{defn}
We call a stationary point $x \in \Delta$ \emph{stable}, if for every 
neighbourhood $U$ of $x$ there is a neighbourhood $V \subset U$ such that
$\Phi(\RR_{\geq 0}, V) \subset U$.
We further call a stationary point $x \in \Delta$ \emph{ 
asymptotically stable} if $x$ is stable and there is a neighbourhood 
$V$ 
of $x$ such that for 
all $y \in V$ we have $\Phi(t,y) \rightarrow x$ for $t \rightarrow \infty$. 
\end{defn}

For stable stationary points we have the following:
\begin{prop}
If $x \in \Delta$ is a stable stationary point of \eqref{eqn:RD}, then $x$ is a 
Nash equilibrium.
\end{prop}

A proof of this statement can be found in \cite[Thm 
5.2]{weibull_evolutionary_1995}.
Note that this further characterization is interesting if
$x \in \partial \Delta$, as stationary points on the boundary of 
$\Delta$ are not necessarily Nash equilibria.
Furthermore, it implies that stationary points that are not Nash 
equilibria must be unstable and thus are harder to attain under 
\eqref{eqn:RD}.
However, note that Nash equilibria do not have to be stable.
We have the following stronger characterization of asymptotically 
stable stationary points (with a proof in, e.g., \cite[Prop. 
5.13]{weibull_evolutionary_1995}):
\begin{prop}
A stationary point $x \in \Delta$ is asymptotically stable under \eqref{eqn:RD} 
if and only if $x$ is a strict Nash equilibrium.
\end{prop}

For completeness, we would like to mention the relationship between stationary 
points of \eqref{eqn:RD} and evolutionarily stable states,
where we define evolutionary stability as in \cite[p. 
166]{weibull_evolutionary_1995}, 
equivalently to \cite{cressman_stability_1992}, as follows:
\begin{defn}[Evolutionary Stability]
We call a state $x^* \in \Delta$ \emph{evolutionarily stable} if for 
all $y \in \Delta$ ($y \neq x^*$) there is some $\bar{\epsilon}_y > 0$ 
such that for all 
$\epsilon \in (0, \bar{\epsilon}_y)$ and $w = \epsilon y + 
(1-\epsilon)x^*$ we have some $i \in I$ with
\[
\bar{f}_i (x_i, w_{-i}) > \bar{f}_i (y_i, w_{-i}) .
\]
\end{defn}

It is well known that in the multi-population case the concept of 
evolutionary 
stability is equivalent to that of a strict Nash equilibrium, e.g., 
\cite[Prop. 5.1]{weibull_evolutionary_1995}:
\begin{prop}
$x \in \Delta$ is evolutionarily stable if and only if $x$ is a strict Nash 
equilibrium.
\end{prop}

Therefore, we have that strict Nash equilibria are exactly the evolutionarily 
stable states and exactly the asymptotically stable stationary points of 
\eqref{eqn:RD}.
The dynamics \eqref{eqn:RD} will therefore not have any asymptotically 
stable 
points if the underlying game does not have any strict Nash equilibria.
Furthermore, no mixed Nash equilibrium can be asymptotically stable, 
such that there is no guarantee that any Nash equilibrium is attainable 
under \eqref{eqn:RD} if the game has only mixed Nash equilibria.

\section{Introducing mutation}

We consider the effect of mutation for two reasons.
First, the idea of evolution is intricately linked with mutation and 
mutation does not seem to be an extraordinary event but is to be 
expected.
Second, a central idea in the proof that the dynamics \eqref{eqn:RD} 
has 
no interior asymptotically stable states relies on the fact that 
\eqref{eqn:RD} is divergence free (after suitable modification) and 
therefore volume preserving, \cite{hofbauer_evolutionary_1998}.
However, some games, such as the matching pennies game and the standard 
rock-paper-scissors game, have only interior equilibria, while 
describing biologically relevant interspecies interactions such as 
host-parasite systems.
The kind of mutation we consider results quite clearly in a dynamics 
with negative divergence.
Of course, this does not guarantee asymptotically stable interior 
equilibria, but it opens up the possibility of such equilibria.

We will first give a motivational heuristic derivation of our specific 
replicator-mutator dynamics from a more general form. Afterwards, we will 
consider the properties of our specific dynamics and of its equilibria.

\subsection{Replicator-Mutator Dynamics}

\subsubsection*{General mutation}
In the standard replicator dynamics \eqref{eqn:RD}, we assume that the 
offspring of individuals of some type inherit that same type.
In contrast, we consider mutation as a process by which the offspring of a 
certain individual changes into another type (of the same population) with some 
probability.
More precisely, we assume that the offspring of an $h $-type in population $i $ 
mutates to a $k $-type in the same population
with some probability
$\mu_{i k h} > 0$,
with 
$\sum_{k \leq n_i } \mu_{i k h} = 1$ for all populations $i $,
and therefore:
\[ \mu_{i h h } = 1 - \sum_{k \neq h } \mu_{i k h}\]

In order to represent overall mutation more clearly, we introduce 
\emph{relative mutation probabilities} $c_{i k h}$ and an overall 
mutation 
rate 
$\mu_i $ such that $\mu_{i k h} = \mu_i c_{i k h}$ ($h \neq k $) and 
thus:
\[
   \mu_{i h h} = 1 - \mu_i \sum_{k \neq h} c_{i k h}
\]
Here, $\mu_i $ controls the overall strength of mutation, such that for $\mu_i 
= 
0$ there is no mutation at all, without affecting relative probabilities.
We derive our specific dynamics from the general multi-population 
replicator-mutator dynamics as given in, e.g., 
\cite{page_unifying_2002},
\begin{align}
\dot{x}_{i h } & = \sum_{k \leq n_i } \mu_{i h k} x_{i k } f_{i k }(x) 
- x_{i h } \bar{f}_i (x) \label{eqn:genRMD}
\end{align}
yielding after substitution:
\begin{align}
\dot{x}_{i h } &
= x_{i h } (f_{i h }(x) - \bar{f}_i (x)) + \mu_i \sum_{k \leq n_i } 
\left( c_{i h k} x_{i k } f_{i k }(x)
 - c_{i k h} x_{i h } f_{i h }(x) \right)  
\end{align}

This formulation emphasizes the similarity to the standard replicator dynamics 
\eqref{eqn:RD} and how $\mu_i $ determines the extent to which 
\eqref{eqn:genRMD} deviates from \eqref{eqn:RD}.

\subsubsection*{Weak selection-weak mutation limit}
Recall that \eqref{eqn:RD} is invariant under the addition of a 
background fitness for all types of a population, a property which 
\eqref{eqn:genRMD} does not have.
We therefore derive a version which is invariant under the addition of 
a constant background fitness.
For convenience, let ${s_i }^{-1}$ denote some background fitness.%
\footnote{Here, $s_i $ can be 
seen as representing the selection pressure on that particular trait.} %
Formulating \eqref{eqn:genRMD} with a modified fitness function 
$\tilde{f}_{i h } : x \mapsto f_{i h }(x) + {s_i }^{-1}$ and suitable 
substitution yields a dynamics with explicit background fitness:
\begin{align*}
\dot{x}_{i h }
& = 
\phi_{ih}(x)
+ \frac{ \mu_{i } }{s_i } \sum_{k \leq n_i } \big( s_i \left(
     c_{i h k} x_{i k } f_{i k }(x)
    - c_{i k h} x_{i h } f_{i h }(x)
  \right)
+ 
    c_{i h k} x_{i k }
    - c_{i k h} x_{i h } \big)
\end{align*}

Analogous to \cite{hofbauer_evolutionary_1998}, we consider a weak 
selection-weak mutation limit, where the background fitness tends to 
infinity, i.e., the selection pressure goes to zero $s_i \rightarrow 
0$, and mutation occurs on the same order as selection, i.e., 
$\mu_i \rightarrow 0$, such that overall:
\[
    \frac{\mu_i }{s_i } \rightarrow M_i > 0
\]
This yields the following weak selection-weak mutation limit of 
\eqref{eqn:genRMD}, which is invariant under addition of background 
fitness,
\begin{align}
\dot{x}_{i h }
& = 
x_{i h } ( f_{i h }(x) - \bar{f}_i (x) )
+ M_i \sum_{k \leq n_i } ( c_{i h k} x_{i k } - c_{i k h} x_{i h } )
\label{eqn:weakRMD}
\end{align}
where we refer to $M_i $ as the \emph{mutation rate} in population $i$.
Note that \eqref{eqn:weakRMD} can also be derived from a discrete 
selection-mutation equation, \cite{hofbauer_evolutionary_1998}.

Additionally, we assume that the mutation probabilities only depend on 
the target type, i.e., ${c_{i h k} = c_{i h l}}$ ($\forall i , h , k 
\neq h , l \neq h $), and we can write $c_{i h}$ instead of $c_{i h k}$ 
and that the mutation rate is the same for every population, 
replacing $M_i $ with $M$,
resulting in the following:%
\footnote{
Although we consider $M$ as independent of the population, population-dependent 
mutation parameters $M_i$ are mostly compatible with the present arguments, but 
would render proofs overly technical.
}
\subsubsection*{Replicator-Mutator Dynamics}
For some fixed $c \in \setint{\Delta}$ and $M \geq 0$, the 
replicator-mutator dynamics \eqref{eqn:RMD} is given by:
\begin{gather}\label{eqn:RMD}\tag{RMD}
\dot{x}_{ih} = \phi^M_{ih}(x) := x_{ih} ( f_{ih}(x) - \bar{f}_i (x) ) + 
M ( c_{ih} - x_{ih} )
\end{gather}

It is clear that we obtain \eqref{eqn:RD} for $M = 0$.
We denote by $\Phi^M: \RR \times \Delta \rightarrow 
\Delta$ the flow of \eqref{eqn:RMD}, i.e., for $x \in \Delta$, 
$\Phi^M(\cdot, x): \RR \rightarrow \Delta, t \mapsto \Phi^M(t,x)$ is a 
solution of 
\eqref{eqn:RMD} with $\Phi^M(0,x)=x$.

\begin{remark*}
Note that $\Phi^M$ also depends on our choice of $c$.
Throughout this section, we will consider some arbitrary but 
\emph{fixed} $c \in 
\setint{\Delta}$ and the defined concepts will depend on that choice. 
However, we will proceed to properties of \eqref{eqn:RMD} which are invariant 
under the choice of $c$ later on.
\end{remark*}

\begin{defn}
We call $x \in \Delta$ with $\phi^M(x) = 0$ a \emph{mutation equilibrium} for 
$M$. For shortness, we call $x^M$ a mutation equilibrium if it is a mutation 
equilibrium for $M$. 
\end{defn}

\begin{defn}
We call a sequence $(x_n)_{n \in \NN} \subset \Delta$ a sequence of 
mutation equilibria if there is a sequence $(M_n)_{n \in \NN} \subset 
\RR_{> 0}$ with
\begin{enumerate}[i)]
\item $M_n \rightarrow 0$ for $n \rightarrow \infty$
\item and $x_n$ is a mutation equilibrium for $M_n$, i.e., $\phi^{M_n}(x_n) = 
0$, 
for all $n \in \NN$.
\end{enumerate}
For ease of notation, we write such a sequence as $(x^M)_{M > 0}$.
\end{defn}

Under suitable assumptions, such sequences represent the change of a 
coevolutionary system under decreasing mutation rates, and we will be 
especially interested in the limits of such sequences of mutation 
equilibria and in their properties.

\subsection{Existence of stationary points with mutation}

\begin{lemma}
For all $M > 0$ and $c \in \setint{\Delta}$ there is $x \in \setint{\Delta}$, 
such that $x$ is a stationary point of the replicator-mutator dynamics 
\eqref{eqn:RMD}, i.e., $\phi^M(x) = 0$.
\end{lemma}
\begin{proof}

Note that the vector field $\phi^M$ points towards the interior of $\Delta$ for 
all $x \in \partial \Delta$.

We thus have that for all $x \in \partial \Delta$ and all $t > 0$, 
$\Phi^M(t,x) \in \setint{\Delta}$, and thus $\Delta$ is forward-invariant under 
the flow $\Phi^M$, in particular, $\Phi^M(\RR_{> 0}, \Delta) \subset 
\setint{\Delta}$.
Furthermore, it is clear that $\Delta$ is nonempty, convex and compact.
Using Brouwer's fixed point theorem, we can now use that if a nonempty, 
convex compact set is forward-invariant under a flow, then it contains 
a fixed point, e.g., \cite[Lemma 6.8]{teschl_ordinary_2012}.
With $\Phi^M(\RR_{> 0}, \Delta) \subset \setint{\Delta}$, we have that the 
fixed point has to be in $\setint{\Delta}$.
\end{proof}

The following definition, e.g., as given by 
\cite{fudenberg_limit_1986}, will be useful in our later investigation:
\begin{defn}[$\epsilon$-Equilibrium]
For some $\epsilon > 0$, we call a state $x^\epsilon \in \Delta$ an 
\emph{$\epsilon$-equilibrium} if
\[
\forall i \in I, h \leq n_i :\: f_{ih}(x^\epsilon) - \bar{f}_i(x^\epsilon) \leq 
\epsilon \,.
\]
\end{defn}

In relation to $\epsilon$-equilibria we state the following property:
\begin{lemma}
Let $x^M$ be a mutation equilibrium, then $x^M$ is an $\epsilon$-equilibrium of 
the underlying game for $\epsilon = M$, in particular:
\[
\forall i \in I, h \leq n_i :\: f_{ih}(x^M) - \bar{f}_i(x^M) < M
\]
\end{lemma}
\begin{proof}
For $(i,h) \in S$, we have that
\begin{align*}
0 = \phi^M_{ih}(x^M)
 = &
x^M_{ih} ( f_{ih}(x^M) - \bar{f}_i(x^M) ) + M (c_{ih} - x^M_{ih})
\\
 > &
x^M_{ih} ( f_{ih}(x^M) - \bar{f}_i(x^M) ) - M x^M_{ih}
\end{align*}
and thus, with $x^M \in \setint{\Delta}$:
\[
f_{ih}(x^M) - \bar{f}_i(x^M) < M
\]
\end{proof}

Together with the continuity of $f$, we have the following:
\begin{coroll}
Let $(x^M)_{M > 0}$ be a sequence of mutation equilibria and $x^*$ an 
accumulation point for $M \rightarrow 0$.
Then $x^*$ is a Nash equilibrium.
\end{coroll}

\subsection{Mutation equilibria for high mutation rates}

We consider some specific properties under high mutation rates which 
illustrate the effect of mutation on the number and stability of 
equilibria through its effect on the Jacobian of the replicator 
dynamics.
Note that all equilibria of \eqref{eqn:RMD}, irrespective of the 
specific choice of $M > 0$, lie in the interior of $\Delta$ and that 
$\phi^M$ points inward on $\partial \Delta$.
We can therefore consider \eqref{eqn:RMD} as a dynamics on 
$\setint{\Delta}$.
We can further, for all populations $i$, replace $x_{i n_i}$ with 
$\left(1 - \sum_{k < n_i} x_{i k}\right)$, and thus proceed to the 
resulting reduced system $\tilde{\phi}^M$, which is then defined on the 
Cartesian product of the $(n_i-1)$-simplices.
For ease of notation, we will still use $\Delta$ to denote this reduced 
space.
Thus, questions regarding the stability of a mutation 
equilibrium $x^M \in \setint{\Delta}$ can be treated by considering the 
eigenvalues of the Jacobian $D\tilde{\phi}^M$.
In particular, due to the Hartman-Grobman theorem, e.g., 
\cite{perko_differential_2001, teschl_ordinary_2012}, we have the 
following useful characterization:
\begin{remark}\label{rem:asympt}
Let $x^M$ be a \emph{hyperbolic} equilibrium of \eqref{eqn:RMD}, and of 
the reduced system $\tilde{\phi}^M$ equivalently, i.e., all eigenvalues 
of $D \tilde{\phi}^M(x^M)$ have non-zero real part.
Then $x^M$ is asymptotically stable if and only if all eigenvalues of 
$D \tilde{\phi}^M(x^M)$ have negative real part, e.g., \cite[Thm 
6.10]{teschl_ordinary_2012}.
In particular, all eigenvalues of $D \tilde{\phi}^M(x^M)$ have negative 
real part, if and only if all eigenvalues of $D\tilde{\phi}(x^M)$ have 
real part smaller than $M$, due to $D\tilde{\phi}^M = D\tilde{\phi} - 
M \cdot Id$.
\end{remark}

With this observation, we obtain the following:

\begin{lemma}\label{lem:highMutation}
There is $\underline{M} \geq 0$ such that for all $M > \underline{M}$ 
the stationary points of the replicator-mutator dynamics 
\eqref{eqn:RMD} are asymptotically stable.
In particular, $D{\tilde{\phi}^M}$ is regular everywhere on $\Delta$.
\end{lemma}
\begin{proof}
Note that all eigenvalues of $D\tilde{\phi}$ are bounded on $\Delta$, 
in particular the real parts of the eigenvalues are bounded, as well.
Then let $\underline{M}$ be an upper bound on all real parts of the 
eigenvalues of $D\tilde{\phi}$ on $\setint{\Delta}$, i.e.:
\[
\underline{M} = \sup\, \{ \Re(\lambda) \;|\; \lambda \in 
\sigma(D\tilde{\phi}(x)), x \in \Delta 
\}
\]
Let $x^M \in \setint{\Delta}$ be a mutation equilibrium for some
$M > \underline{M}$.
As noted, the Jacobian of $\tilde{\phi}^M$ satisfies
$ D{\tilde{\phi}^M}(x) = D{\tilde{\phi}}(x) - M \cdot Id $ for all $x 
\in \Delta$.
In particular, for all eigenvalues
$\lambda^M \in \sigma( D{\tilde{\phi}^M}(x^M) )$ we have that
$\lambda^M + M \in \sigma( D{\tilde{\phi}}(x^M) )$ and hence
$\Re(\lambda^M) + M \leq \underline{M}$, and thus $\Re(\lambda^M) < 0$.
Therefore, all eigenvalues of $D{\tilde{\phi}^M}(x^M)$ have strictly 
negative real parts and with remark \ref{rem:asympt}, $x^M$ is 
asymptotically stable.
\end{proof}

\begin{remark*}
Note that that the $\underline{M}$ in the previous lemma 
\ref{lem:highMutation} is independent of the choice of $c \in 
\setint{\Delta}$, thus giving a lower bound on the mutation rate above 
which all equilibria are asymptotically stable independent of $c \in 
\setint{\Delta}$.
\end{remark*}

\subsubsection{Uniqueness of mutation equilibria for high mutation rates}

For very high mutation ($M > \underline{M}$) we further obtain that mutation 
equilibria are unique and that there is a continuously differentiable function 
mapping mutation rates to mutation equilibria.
We first consider the following lemma:

\begin{lemma}\label{lemma:highMutationContFunc}
Let $c \in \setint{\Delta}$ and $\underline{M}$ from lemma 
\ref{lem:highMutation}.
Let $x^M$ be a mutation equilibrium for some $M > \underline{M}$.
Then there is a unique function
$\mathcal{M}: (\underline{M}, \infty) \rightarrow \Delta$
such that $\mathcal{M}(M) = x^M$ and for all $m \in (\underline{M}, 
\infty)$, $\mathcal{M}(m)$ is a mutation equilibrium for $m$.
In particular, $\mathcal{M}$ is continuously differentiable and 
$\mathcal{M}(m) \overset{m \rightarrow \infty}{\longrightarrow} c$.
\end{lemma}
This is denoted corollary \ref{coroll:highMutationContFunc} in the 
appendix where the proof is given. \\

Note that this does not guarantee any uniqueness of equilibria, yet, 
only the uniqueness of functions passing through a given equilibrium.
The uniqueness of mutation equilibria for high mutation rates is then 
obtained in the next step from the fact that we have uniqueness at 
least for some mutation rate:
\begin{prop}\label{prop:uniquenessForHighMutation}
Let $c \in \setint{\Delta}$ and $\underline{M}$ from lemma 
\ref{lem:highMutation}.
For all $M > \underline{M}$, the replicator-mutator dynamics 
\eqref{eqn:RMD} has a unique mutation equilibrium.
The unique map $\mathcal{M}: M \mapsto x^M$ is continuously 
differentiable on $(\underline{M}, \infty)$.

\end{prop}
This is denoted proposition \ref{appProp:uniquenessForHighMutation} in 
the appendix where the proof is given.

\begin{remark}
Note that the main achievement of proposition 
\ref{prop:uniquenessForHighMutation} is to extend the 
uniqueness of equilibria beyond any Lipschitz constant of 
$\tilde{\phi}$ to $(\underline{M}, \infty)$, i.e., to the interval 
where $D\tilde{\phi}^M$ is guaranteed to be regular.
Furthermore, if $D\tilde{\phi}^M(x^M)$ is invertible for all 
$M \in (a,\infty)$ and corresponding mutation equilibria $x^M$ then the 
uniqueness extends to $(a, \infty)$. In fact, if $a = 0$ then there is 
a unique sequence of mutation equilibria $(x^M)_{M > 0}$ for $c \in 
\setint{\Delta}$ since it is induced by the function $\mathcal{M}$.
\end{remark}

For a fixed $c \in \setint{\Delta}$ and a sufficiently high 
mutation rate, the unique mutation equilibrium will be arbitrarily 
close to $c$.
Therefore, if we were interested in finding the mutation equilibrium 
for a sufficiently high mutation rate, we could choose an initial point 
close to $c$ and the dynamics \eqref{eqn:RMD} would converge to the 
asymptotically stable mutation equilibrium.
The uniqueness on $(\underline{M}, \infty)$ further enables us to lower 
the mutation rate almost to $\underline{M}$ without losing uniqueness 
and asymptotic stability.

\section{Mutation limits}

In our previous considerations, we assumed fixed relative mutation 
probabilities $c \in \setint{\Delta}$.
In particular, certain effects could depend on the specific choice of $c$, 
e.g., if we picked $c$ to coincide with a Nash equilibrium of the underlying 
game.
However, we are interested in properties that are independent of the specific 
choice of $c$.
To this end, we introduce the following definition:

\begin{defn}[Mutation Limit]
We call a connected compact set $X \subset \Delta$ a 
\emph{mutation limit}, if for all $c \in \setint{\Delta}$ there is a 
sequence of mutation equilibria $(x^M)_{M > 0} \subset \Delta$ that 
converges to an element of $X$ and $X$ contains no proper subset with 
these properties.
We call $x \in \Delta$ a \emph{mutation limit (point)} if the singleton 
set $\{x\}$ is a mutation limit.
\end{defn}

\begin{remark*}
It is clear that every mutation limit $X \subset \Delta$ is a subset of 
the set of Nash equilibria, $\nashEq/$, of the underlying game, as 
the limit of \emph{any} sequence of mutation equilibria is a Nash 
equilibrium.
Furthermore, if it exists, it must be contained in a connected 
component of $\nashEq/$.
\end{remark*}

\subsection{General existence of mutation limits}

A question that arises from the definition is that of the existence of 
mutation limit points.
While we have shown that for any fixed $c \in \setint{\Delta}$ and any 
mutation rate $M > 0$ there is a corresponding mutation equilibrium and 
therefore the Bolzano-Weierstrass theorem guarantees the existence of a 
limit for vanishing mutation, this limit need not be independent of the 
choice of $c$, and indeed it could be possible that there is no 
mutation limit at all, neither a singleton set nor otherwise.
The question, therefore, is whether every game has at least one 
mutation limit point.
To this question, we can give a negative answer, as the following 
example shows:
\begin{exmpl}
Consider a two-player game with the following payoff structure:
\begin{center}
\begin{tabular}{c|cc}
		&	$C_1$	&	$C_2$	\\ \hline
$R_1$	&	1, 0	&	0, 1	\\
$R_2$	&	0, 1	&	1, 0	\\
$R_3$	&	0, 1	&	1, 0	\\
\end{tabular}
\end{center}

It is clear that any Nash equilibrium of the game has the form
\[ \left( \left(\frac{1}{2},\, \frac{t}{2},\, \frac{1-t}{2} \right), 
\left(\frac{1}{2},\, \frac{1}{2} \right) \right) \]
with $t \in [0,1]$, where we give the strategy of the row player first.
Excluding a few special choices of $c \in \setint{\Delta}$, for 
any generic $c$ given as $(c_{R,1}, c_{R,2}, 
c_{R,3}, c_{C,1}, c_{C,2})$, every sequence of mutation equilibria will 
converge to a Nash equilibrium of the above form with:
\[ t = \frac{ c_{R,2} }{c_{R,2} + c_{R,3}} \]
It is therefore evident that this game has no mutation limit point, 
i.e., there is no Nash equilibrium that is approached by mutation 
equilibria for all choices $c \in \setint{\Delta}$.
However, for any Nash equilibrium $x$ of the above form with $t \in (0,1)$ 
there is a $c \in \setint{\Delta}$ such that $x$ is approached by a 
sequence of mutation equilibria. Therefore, the set of Nash equilibria 
is indeed a mutation limit.
\end{exmpl}

In the above example, the set of all Nash equilibria turns out to be a 
mutation limit.
However in general, the set of Nash equilibria need not be connected.
In this context, the following result answers the question about the 
general existence of mutation limits:
\begin{prop}\label{prop:mutationlimitsingames}
For every $f \in \contdiff{1}(U \supset \Delta, \RR^S)$ there is a 
mutation 
limit.
\end{prop}
\begin{proof}
See appendix \ref{sec:proof:mutationlimitsingames}.
\end{proof}

Note that this result does not require that there is no intraspecies 
interaction, i.e., it does not require $\frac{\partial}{\partial 
x_{ik}} f_{ih}(x) = 0$ ($\forall i \in I, h,k \in S_i, x \in \Delta$).
In fact, the proof can be quite easily generalized to other, not 
necessarily replicator dynamics.

From proposition \ref{prop:mutationlimitsingames}, together with the 
prior remark that a mutation limit is contained in a connected 
component of $\nashEq/$, 
we obtain the following existence result for dynamics with only a 
finite number of Nash equilibria:
\begin{coroll}
Let $f \in \contdiff{1}(U \supset \Delta, \RR^S)$ such that the set of 
Nash 
equilibria, $\nashEq/$, is finite.
Then all mutation limits are mutation limit points and there is at 
least one mutation limit point.
\end{coroll}

Note that the finiteness condition is particularly important for 
fitness functions that are not derived from finite normal-form games.

\subsubsection{A sufficient condition for mutation limits}

We can further guarantee that regular Nash equilibria, introduced in 
\cite{harsanyi_oddness_1973}, cf. also \cite{van_damme_stability_1991}, 
are mutation limit points, where we employ the following equivalent 
definition, \cite{ritzberger_theory_1994}:

\begin{defn}
We call a Nash equilibrium $x \in \Delta$ a \emph{regular equilibrium} 
if the reduced Jacobian of \eqref{eqn:RD} at $x$, $D\tilde{\phi}(x)$, 
has full rank.
\end{defn}

In particular, all strict Nash equilibria are regular, \cite[Cor. 
2.5.3]{van_damme_stability_1991}.

\begin{lemma}
Let $x^*$ be a regular equilibrium. Then $x^*$ is a mutation limit, i.e.,
for all $c \in \setint{\Delta}$, there is a sequence of mutation 
equilibria, $(x^M)_{M > 0}$, such that $x^M \rightarrow x^*$ for $M \rightarrow 
0$.
\end{lemma}
\begin{proof}
Note that $D\tilde{\phi}(x^*)$ is invertible and therefore, by the 
implicit function theorem, for every $c \in \setint{\Delta}$, there is 
a continuously differentiable $\mu: (-\epsilon, \epsilon) 
\rightarrow \RR^N$ for some $\epsilon > 0$, such that for $M \in 
(-\epsilon, \epsilon)$ we have that $\tilde{\phi}^M(\mu(M)) = 0$.
Of course, negative values of $M$ are not interpretable as mutation 
rates and we consider them here only for technical reasons of 
differentiability at $0$.

If $x^* \in \setint{\Delta}$, then it is clear that we can choose $\epsilon$ 
such that $\mu( [0,\epsilon] ) \subset \Delta$, and therefore a sequence of 
mutation equilibria $(x^M)_{M > 0} \subset \Delta$ with $x^M \rightarrow x^*$ 
for $M \rightarrow 0$.

Suppose that $x^* \in \partial \Delta$ and for some $(i,h) \in S$ we have 
$x^*_{ih} = 0$.
Note that $\mu$ is continuously differentiable and therefore for
$M \in (-\epsilon, \epsilon)$,
\begin{align*}
0
{} = {}&
\frac{d}{d M} \phi^M_{ih} (\mu(M))
\\
{} = {}&
\frac{d}{d M} \bigg( \mu_{i h} (M) g_{ih} (\mu(M)) \bigg)
+ \frac{d}{d M} \bigg( M (c_{ih} - \mu_{ih}(M)) \bigg)
\\
{} = {}&
g_{ih} (\mu(M)) \frac{d}{d M} \mu_{i h} (M)
+ \mu_{i h} (M) \frac{d}{d M} g_{ih} (\mu(M)) \\
& + (c_{ih} - \mu_{ih}(M)) + M (c_{ih} - \frac{d}{d M} \mu_{ih} (M))
\end{align*}
and hence for $M = 0$,
\begin{align*}
0
{} = {}&
\frac{d}{d M} \phi^M_{ih} (\mu(M))\big|_{M = 0}
\\
{} = {}&
g_{ih} (\mu(0)) \frac{d}{d M} \mu_{i h} (0)
+ \mu_{i h} (0) \frac{d}{d M} g_{ih} (\mu(0)) + (c_{ih} - \mu_{ih}(0)) + 0 
(c_{ih} 
-  \frac{d}{d M} \mu_{ih} (0))
\\
{} = {}&
g_{ih} (x^*) \frac{d}{d M} \mu_{i h} (0)
+ \underbrace{x^*_{ih}}_{= 0} \frac{d}{d M} g_{ih} (x^*) + (c_{ih} - 
\underbrace{x^*_{ih}}_{= 0})
\\
{} = {}&
g_{ih} (x^*) \frac{d}{d M} \mu_{i h} (0) + c_{ih}
> 
g_{ih} (x^*) \frac{d}{d M} \mu_{i h} (0) \;.
\end{align*}
Thus, with $x^*$ being a Nash equilibrium, we have $g_{ih} (x^*) \leq 0$ and 
therefore $\frac{d}{d M} \mu_{i h} (0) \geq 0$. Because of the strict 
inequality, we even have $g_{ih} (x^*) < 0$ and $\frac{d}{d M} \mu_{i h} (0) > 
0$.
Therefore, we can choose $\epsilon$ such that $\mu([0, \epsilon)) \subset 
\Delta$ and a sequence of mutation equilibria converging to $x^*$.
\end{proof}

\begin{remark*}
It should be noted that the proof of the above result shows that there is a 
continuously differentiable function mapping mutation rates to mutation 
equilibria and that this function is unique.
In other words, given a $c \in \setint{\Delta}$, the sequence approaches $x^*$ 
in a unique manner.
\end{remark*}

\subsection{Attracting Mutation Limits}

Up to this point we have considered equilibria (or sets of equilibria) 
of \eqref{eqn:RD} such that for any $c \in \setint{\Delta}$ and 
mutation rate $M>0$ a mutation equilibrium of the respective 
\eqref{eqn:RMD} would be located arbitrarily close, depending on $M$.
We have so far ignored the stability properties of the mutation 
equilibria arising nearby.
If the mutation equilibrium arising nearby happens to be asymptotically 
stable for some mutation rate $M>0$ and some $c \in \setint{\Delta}$, 
then under suitable initial conditions the system will converge to a 
state close to the mutation limit.
However, as with the notion of mutation equilibria, such behaviour of 
the system is mostly of interest if it does not depend on a lucky 
choice of $c$, in particular if nearby mutation 
equilibria turn out to by asymptotically stable for every choice of $c$.
In this case, the mutation limit would be approximated arbitrarily 
close in all \eqref{eqn:RMD} only depending on $M > 0$.
This idea motivates the following formal definition:

\begin{defn}[Attracting Mutation Limit]
We call a mutation limit $X \subset \Delta$ \emph{attracting} if for 
every $c \in \setint{\Delta}$ and every sequence of mutation equilibria 
$(x^M)_{M > 0}$ that converges to an element of $X$, 
there is $m > 0$ such that for all $M < m$, $x^M$ is asymptotically 
stable.
We call $x \in \Delta$ an \emph{attracting} mutation limit (point) if 
the singleton set $\{x\}$ is an attracting mutation limit.
\end{defn}

\subsubsection{A sufficient condition for attracting mutation limits}

It is known that if $x^*$ is a strict Nash equilibrium, then 
$D\tilde{\phi}(x^*)$ has only real, strictly negative eigenvalues, 
e.g., \cite[Lemma 1]{ritzberger_theory_1994}, and $x^*$ is 
therefore regular and thus a mutation limit.
Furthermore, we can show that $x^*$ is an attracting mutation limit:
\begin{lemma}
Let $x^*$ be a strict Nash equilibrium. Then $x^*$ is an attracting 
mutation limit.
\end{lemma}
\begin{proof}
With the previous note, it is clear that $x^*$ is a mutation limit.
It remains to show that the mutation equilibria $(x^M)_{M>0}$ 
converging to $x^*$ for 
any $c \in \setint{\Delta}$ are asymptotically stable. 
Since all eigenvalues of the Jacobian at $x^*$ have strictly negative 
real parts, and in fact are real, \cite{ritzberger_theory_1994},
we have that the eigenvalues of $D\tilde{\phi}(x)$ have strictly 
negative real parts in a neighbourhood of $x^*$, as the roots of a 
polynomial vary continuously with its coefficients, e.g., 
\cite{harris_shorter_1987}, and $D\tilde{\phi}$ is continuous.
Therefore, in a neighbourhood of $x^*$, all eigenvalues of the Jacobian 
of $\tilde{\phi}^M$, with $D\tilde{\phi}^M(x) = D\tilde{\phi}(x) - M 
Id$, have strictly negative real parts for any $M \geq 0$, and thus the 
$x^M$ are asymptotically stable, e.g., \cite{perko_differential_2001}.
\end{proof}

\begin{remark}
Since an equilibrium is strict if and only if it is evolutionarily stable, this 
result also implies that all evolutionarily stable equilibria are attracting 
mutation limits.
\end{remark}

The following example shows that attracting mutation limits are not 
necessarily strict Nash equilibria, and hence that the concept of 
attracting mutation limits is also weaker than evolutionary stability:
\begin{exmpl}
Consider the 2-by-2 matching pennies game given by the payoffs:
\[
\begin{pmatrix}
(1,0) & (0,1) \\
(0,1) & (1,0)
\end{pmatrix}
\]

The strategy profile $((\frac{1}{2}, \frac{1}{2}), (\frac{1}{2}, \frac{1}{2}))$ 
is a Nash equilibrium but not strict.
However, it is an attracting mutation limit, as proven in a forthcoming 
article.
\end{exmpl}

\subsubsection{A necessary condition for attracting mutation limits}

The observation that not all Nash equilibria are attracting mutation limits 
relies on the following:
\begin{lemma}
Let $x^* \in \Delta$ be an attracting mutation limit. Then all eigenvalues of 
the Jacobian $D\tilde{\phi}(x^*)$ have nonpositive real parts.
\end{lemma}
\begin{proof}
Suppose there is an eigenvalue of $D\tilde{\phi}(x^*)$ with a strictly 
positive real part.
Then there is $\epsilon > 0$ and a neighbourhood $U$ of 
$x^*$ such that $D\tilde{\phi}(x)$ has an eigenvalue $\lambda$ with 
$\Re(\lambda) > \epsilon$ for all $x \in U$.
Let $(x^M)_{M > 0}$ be a sequence of mutation equilibria converging to 
$x^*$ for some $c 
\in \setint{\Delta}$.
Then there is $\epsilon'$ such that $x^M \in U$ for $M < \epsilon'$.
In particular, we can choose $\epsilon' < \epsilon$.
Then the Jacobian $D\tilde{\phi}^M(x^M)$, with $D\tilde{\phi}^M(x^M) = 
D\tilde{\phi}(x^M) - M Id$, has an 
eigenvalue with strictly positive real part, and $x^M$ is not asymptotically 
stable, as it is not even stable, e.g., \cite{hirsch_differential_1974}.
Therefore, $x^*$ is not an attracting mutation limit.
\end{proof}

This result, together with the following example, then demonstrates 
that not all Nash equilibria are attracting mutation limits:
\begin{exmpl}
Consider the 2-by-2 coordination game given by:
\[
\begin{pmatrix}
(1,1) & (0,0) \\
(0,0) & (1,1)
\end{pmatrix}
\]
The strategy profile $((\frac{1}{2}, \frac{1}{2}), (\frac{1}{2}, \frac{1}{2}))$ 
is a Nash equilibrium, but its Jacobian has eigenvalues with positive real 
parts and therefore, it is not an attracting mutation limit.
\end{exmpl}

\section{Discussion}

We have shown that a very simple form of mutation leads to qualitative 
changes in the multi-population replicator dynamics.
Furthermore, these changes do not depend on the specific choice of 
parameters but are of a general character.
Not only do mutation limits exist for all continuously differentiable 
fitness functions, mutation can also cause the dynamics to approximate 
equilibria that would not be approximated without mutation, again 
independently of the choice of specific mutation parameters, which is 
due to asymptotically stable equilibria arising close to an original 
equilibrium.
The closest results to our approach that we are aware of are presented 
in \cite{ritzberger_theory_1994}, and if considered as an approximation 
to RD, certain aspects of RMD are clarified by those results, as 
indicated.
The results presented here differ in that they show robustness in 
a system of families of approximations which are not related to 
perturbed normal-form game payoffs and in that they focus on the 
effects on the stability of equilibria, independent of the choice of 
the specific approximation.

With respect to periodic behaviour in biological populations it should 
be noted that the degree of stabilisation of RD depends on the mutation 
rate, resulting in a very slow approach of an asymptotically stable 
mutation equilibrium and seemingly periodic behaviour if mutation is 
low.
In an empirical situation this can lead to difficulties in 
distinguishing dynamics with truly periodic behaviour from ones with 
only seemingly periodic behaviour if measuring on a (relatively) small 
time scale.
Furthermore, in small populations stochastic effects will play a 
significant role.
Therefore, under very low mutation, empirical findings of periodic 
fluctuations can be consistent with our results
if measured in small populations on a small time scale,
such that any stabilising effects of mutation will be more apparent in 
large populations on large time scales, or with sufficiently fast 
reproduction.

On the one hand, given the potential health impacts of even slight 
mutations on organisms and the fact that such mutations occur with a 
non-negligible probability, as mentioned earlier, and given further its 
role as a generator of variety on which evolutionary selection 
operates, it is clear that it is worth including mutation mechanisms in 
the study of populations, and one should expect results that deviate 
potentially significantly from models without mutation.

On the other hand, given that the multi-population replicator dynamics has been 
shown to be related to learning dynamics and that mutation-like terms have been 
shown to arise in formulations of Q-learning algorithms, it is worth noting 
that our results show that replicator-mutator dynamics have more desirable 
convergence properties than the pure replicator dynamics, while remaining 
arbitrarily close to a Nash equilibrium.
Therefore, attracting mutation limits resulting from a 
replicator-mutator dynamics can be considered a more suitable class of 
dynamic solution approaches for games than the pure multi-population 
replicator dynamics.

As shown, attracting mutation limits do not exist for all games, 
and the characterization of their existence is therefore an open 
problem.
We will address this problem partially in forthcoming results on 
attracting mutation limits in the matching pennies game, which can be 
considered a model of antagonistic coevolution.
Furthermore, we have considered a specific form of mutation, and therefore the 
question of which properties carry over to more complicated and more realistic 
mutation mechanisms remains.

\bibliographystyle{abbrv}
\bibliography{db} 

\appendix

\section{Proof of proposition \ref{prop:uniquenessForHighMutation}}

The proof of proposition \ref{prop:uniquenessForHighMutation} relies on 
the implicit function theorem, which we restate for convenience, e.g., 
as in \cite[Thm 3.3.1]{krantz_implicit_2013}:
\begin{thm}[Implicit Function]\label{thm:implFunc}
Let $W \subset \RR$, $X \subset \RR^m$ be open and let
$\rho: W \times X \rightarrow \RR^m, (w,x) \mapsto \rho(w,x)$ be a 
continuously differentiable function.
Let further $(w',x') \in W \times X$ be such that $\rho(w',x') = 0$ and 
the $m \times m$ matrix $ \frac{\partial}{ \partial x} \rho (w',x') $ 
be invertible.

Then there exist an open neighbourhood $W_1 \subset W$ of $w'$, an 
open neighbourhood $X_1 \subset X$ of $x'$, and a continuously 
differentiable function $F: W_1 \rightarrow X_1$ such that $\forall w 
\in W_1: \rho(w,F(w)) = 0$.
Furthermore, for all $(w,x) \in W_1 \times X_1$ we have that $\rho(w,x) 
= 0$ if and only if $x = F(w)$.
\end{thm}

For the proof of proposition \ref{prop:uniquenessForHighMutation} we 
will need a consequence of the implicit function theorem, based on the 
following statement that we can extend an implicitly defined function 
if the conditions of the implicit function theorem hold on the boundary 
of its domain:

\begin{lemma}
Let $\rho: W \subset \RR \rightarrow X \subset \RR^m$ be continuously 
differentiable for $W$ open and $X$ open,
and let $R: W_R \rightarrow X_R$ be continuously differentiable, with 
open and convex $W_R \subset W$ and open $X_R \subset X$, such that:
\begin{enumerate}[i)]
\item  $\forall v \in W_R: \rho(v,R(v)) = 0$;
\item $\forall (v,x) \in W_R \times X_R :\: \rho(v,x) = 0 
\Leftrightarrow x = R(v)$.
\end{enumerate}
If for some sequence $(v_n)_{n \in \NN} \subset W_R$ with
$v_n \rightarrow v' \in \partial W_R \cap W$ and an accumulation point 
$x' \in X$ of $(R(v_n))_{n \in \NN}$, 
the matrix $ \frac{\partial}{ \partial x} \rho (v',x') $ is invertible,
then there is a unique continuously differentiable extension of $R$ 
with the above properties whose domain is open and a proper superset of 
$W_R$.
In particular, $(R(v_n))_{n \in \NN}$ is convergent with limit $v'$.
\end{lemma}
\begin{proof}
Let $(v_n)_{n \in \NN} \subset W_R$ with $v_n \rightarrow v' 
\in \partial W_R \cap W$ and let $x' \in X$ be an 
accumulation point of $(R(v_n))_{n \in \NN}$,
such that the matrix $ \frac{\partial}{ \partial x} \rho (v',x') $ 
is invertible.
Due to the continuity of $\rho$ on $W \times X$, we have that $\rho(v', 
x') = 0$.
With the implicit function theorem, there are open neighbourhoods $W' 
\subset W$ of $v'$, where we can require $W'$ to be convex, and $X' 
\subset X$ of $x'$ and a unique continuously differentiable function 
$S: W' \rightarrow X'$ with the corresponding properties i) and ii).

We will show that there is $N$ such that $(R(v_n))_{n \geq N} \subset 
X'$:
As $x'$ is an accumulation point of $(R(v_n))_{n \in \NN}$,
there are infinitely many $n \in \NN$ with $R(v_n) \in X'$,
in particular let $R(v_N) \in X'$.
Note that we can assume $(v_n)_{n \geq N} \subset W'$ as $v' \in W'$ is 
the limit of that sequence.
Assume that there is some $N' > N$ with $R(v_{N'}) \notin X'$ and 
let $N'$ be minimal.
W.l.o.g. let $N' = N + 1$ and define $v: [0,1] \rightarrow W', t 
\mapsto (1-t) v_{N} + t v_{N'}$.
Then $v([0,1]) \subset W'$ due to convexity.
Consider that $R(v_{N}) \in X'$, with $X'$ open.
Therefore, there is some $\epsilon > 0$ with $R(v([0,\epsilon])) 
\subset X'$.
However, with our assumption, $R(v(1)) = R(v_{N'}) \notin X'$. 
Then, with the complement of $X'$ being closed, there is a minimal 
$\bar{t}$ such that $R(v(\bar{t})) \notin X'$.
Then $R \circ v = S \circ v$ on $[0, \bar{t})$, but due to their 
continuity we then also have $R(v(\bar{t})) = S(v(\bar{t}))$ and thus 
$R(v(\bar{t})) \in X'$, in contradiction to $R(v(\bar{t})) \notin X'$.
Thus, $R(v_{N'}) = R(v(1)) \in X'$, in contradiction to 
$R(v_{N'}) \notin X'$.
Overall, we then have $(R(v_n))_{n \geq N} \subset X'$,
and further $R([v_N, v')) \subset X'$ (assuming $v_N < v'$).
This implies that $R = S$ on $W_R \cap W'$ and $T := R \cup S$ is a 
proper, continuously differentiable extension of $R$, with 
corresponding properties i) and ii).
In particular, due to $(R(v_n))_{n \geq N} = (T(v_n))_{n \geq N}$, 
$(R(v_n))_{n \in \NN}$ is convergent with limit $v'$.
\end{proof}

The following lemma states that there is an implicitly defined function 
whose domain is such that the points at the boundary do not satisfy the 
conditions of the implicit function theorem:

\begin{lemma}
Let $\rho: W \rightarrow X$ be as given in the implicit function 
theorem and $(w,x^w) \in W \times X$ such that $\rho(w,x^w) = 0$ and 
the matrix $ \frac{\partial}{ \partial x} \rho (w,x^w) $ is invertible.
Then there exist open neighbourhoods $W^* \subset W$ of $w$, with 
$W^*$ convex, and $X^* \subset X$ of $x^w$, and a continuously 
differentiable function $R^*: W^* \rightarrow X^*$ such that:
\begin{enumerate}[i)]
\item  $\forall v \in W^*: \rho(v,R^*(v)) = 0$;
\item $\forall (v,x) \in W^* \times X^* :\: \rho(v,x) = 0 
\Leftrightarrow x = R^*(v)$;
\item for all $(v_n)_{n \in \NN} \subset W^*$ with $v_n \rightarrow v' 
\in \partial W^* \cap W$ and every accumulation point $x' \in X$ of 
$(R^*(v_n))_{n \in \NN}$, the matrix $ \frac{\partial}{ \partial x} 
\rho (v',x') $ is singular.
\end{enumerate}
In particular, $R^*$ is a maximally defined such function.
\end{lemma}
\begin{proof}

Let $\mathcal{R}$ be the set of all continuously differentiable 
functions $R_\alpha: W_\alpha \rightarrow X_\alpha $,
with $W_\alpha \subset W$ convex and $X_\alpha \subset X$ being open 
neighbourhoods of $w$ and $x^w$, respectively, such that
$R_\alpha$ satisfies i) and ii).
Due to $\rho$ being continuously differentiable,
$ \frac{\partial}{ \partial x} \rho $ is regular in a convex, 
open neighbourhood of $(w, x^w)$.
With the implicit function theorem, $\mathcal{R}$ is not empty.
We define a partial order on $\mathcal{R}$ by the set inclusion on the 
graphs of the functions $R_\alpha \in \mathcal{R}$.

Let $\mathcal{O}$ be a non-empty completely ordered chain in 
$\mathcal{R}$.
Consider the function $R'$ defined by the graph:
\[
\Gamma(R') =
  \bigcup_{R_\alpha \in \mathcal{O}} \{ (v, R_\alpha(v)) \,|\, v \in 
  W_\alpha \}
\]
Then
$W' = \bigcup_{R_\alpha \in \mathcal{O}} W_\alpha \subset W$ and
$X' = \bigcup_{R_\alpha \in \mathcal{O}} X_\alpha \subset X$ are 
open neighbourhoods of $w$ and $x^w$ and $R': W' \rightarrow X'$ is a 
continuously differentiable function.
Furthermore, $\{ W_\alpha \,|\, R_\alpha \in \mathcal{O} \}$ is 
completely ordered by set inclusion as well and therefore, $W'$ is 
convex.
It is clear that $R'$ satisfies i) as all $R_\alpha$ satisfy i).
Let $(v,x) \in W' \times X'$.
Then there is $R_\alpha \in \mathcal{O}$ with $v \in W_\alpha$, $x \in 
X_\alpha$, and $R'(v) = R_\alpha(v)$.
Then, as $R_\alpha$ satisfies ii), we have $\rho(v,x) = 0 
\Leftrightarrow x = R_\alpha(v) = R'(v)$, and thus $R'$ 
satisfies ii).
Therefore, $R' \in \mathcal{R}$, and with Zorn's Lemma, $\mathcal{R}$ 
contains a maximal element $R^*: W^* \rightarrow X^*$, such that $R^*$ 
satisfies i) and ii).

For iii), let $(v_n)_{n \in \NN} \subset W^*$ with $v_n \rightarrow v' 
\in \partial W^* \cap W$ and let $x' \in X$ be an accumulation point of 
$(R^*(v_n))_{n \in \NN}$.
In particular, this implies $\overline{W^*} \subset W$.
Assume that the matrix $ \frac{\partial}{ \partial x} \rho (v',x') $ 
is invertible.
With the previous lemma 
there is a proper extension of $R^*$ and $R^*$ is not maximal, a 
contradiction.
Therefore, the matrix $ \frac{\partial}{ \partial x} \rho (v',x') $ 
is singular.
\end{proof}

In order to apply the above lemma, for $M > 0$, we rewrite 
\eqref{eqn:RMD} as
\begin{gather}\label{eq:rho}
\rho: \RR \times X \rightarrow \RR^S, (w,x) \mapsto w \phi(x) + (c 
- x)
\end{gather}
with $w = M^{-1}$.
It is clear that $\rho(M^{-1},x) = M^{-1} \phi^{M}(x)$ and 
therefore $\rho(M^{-1},x) = 0 \Leftrightarrow \phi^{M}(x) = 0$ and that 
$\rho$ is continuously differentiable on $\RR \times X$ with some
$X \supset \Delta$ open and bounded, depending on $\phi$.
Then we obtain lemma \ref{lemma:highMutationContFunc} as a corollary:
\begin{coroll}\label{coroll:highMutationContFunc}
Let $c \in \setint{\Delta}$ and $\underline{M}$ be as in lemma 
\ref{lem:highMutation}.
Let $x^M$ be a mutation equilibrium for some $M > \underline{M}$.
Then there is a unique function
$\mathcal{M}: (\underline{M}, \infty) \rightarrow \Delta$
such that $\mathcal{M}(M) = x^M$ and for all $m \in (\underline{M}, 
\infty)$, $\mathcal{M}(m)$ is a mutation equilibrium for $m$.
In particular, $\mathcal{M}$ is continuously differentiable and 
$\mathcal{M}(m) \overset{m \rightarrow \infty}{\longrightarrow} c$.
\end{coroll}
\begin{proof}
Consider that for $m > \underline{M}$, $D{\phi^M}$ is invertible 
everywhere on $\Delta$ due to lemma \ref{lem:highMutation}, and that 
for 
$w = m^{-1}$ with $\rho$ from \eqref{eq:rho}, the matrix
$ \frac{\partial}{ \partial x} \rho (w,x) $ is invertible whenever 
$D{\phi^m}(x)$ is.
Then let $\underline{w} = \underline{M}^{-1}$ and $w = M^{-1}$ for some 
$M > \underline{M}$.
Then applying the previous lemma to $w$, $x^M$ and $\rho$ yields a 
continuously differentiable function
$R: W \rightarrow \Delta$ with $W \subset \RR$ and $w \in W$.
Furthermore, the previous lemma guarantees that
$[0, \underline{w}) \subset W$ because
$ \frac{\partial}{ \partial x} \rho (v,R(v)) $ is invertible 
$\forall v \in [0, \underline{w})$.
Thus, $\mathcal{M}: (\underline{M}, \infty) \rightarrow \Delta$ with $m 
\mapsto R(m^{-1})$ is continuously differentiable and has the 
desired properties.
\end{proof}

With this we can prove proposition \ref{prop:uniquenessForHighMutation}:
\begin{prop}[\ref{prop:uniquenessForHighMutation}]\label{appProp:uniquenessForHighMutation}
Let $c \in \setint{\Delta}$ and $\underline{M}$ as in lemma 
\ref{lem:highMutation}.
For all $M > \underline{M}$, the replicator-mutator dynamics 
\eqref{eqn:RMD} has a unique mutation equilibrium.
The unique map $\mathcal{M}: M \mapsto x^M$ is continuously 
differentiable on $(\underline{M}, \infty)$.

\end{prop}
\begin{proof}
As $\phi$ is Lipschitz, let $L_\phi$ be the best Lipschitz constant for $\phi$.
Since $\phi$ is differentiable and $\Delta$ is convex, we further have 
that
\[
L_\phi = \| D\phi \|_{\infty, \Delta} := \sup_{x \in \Delta} \| 
D\phi(x) \| 
\geq \underline{M} 
\]
with $\underline{M}$ from lemma \ref{lem:highMutation}. Choose $M' > 
L_\phi$ and consider for $c \in \setint{\Delta}$ and some $s > 0$ the 
function $F_{M',c}: \Delta \rightarrow 
\Delta$ with:
\[
[ F_{M',c}(x) ]_{ih} = x_{ih} + s \left( \phi_{ih}(x) + M' (c_{ih} - 
x_{ih}) 
\right)
\]
Then, we have that
\begin{align*}
[ F_{M',c}(x) ]_{ih} - [ F_{M',c}(y) ]_{ih}
= &
x_{ih} + s \left( \phi_{ih}(x) + M' (c_{ih} - x_{ih}) \right)
\\
&
- y_{ih} - s \left( \phi_{ih}(y) + M' (c_{ih} - y_{ih}) \right)
\\
= &
(1 - s M') (x_{ih} - y_{ih}) + s \left(
	 \phi_{ih}(x) - \phi_{ih}(y) 
\right)
\end{align*}
and thus
\begin{align*}
\| F_{M',c}(x) - F_{M',c}(y) \|
\leq &
| 1 - s M' | \| x - y \| + s \| \phi(x) - \phi(y) \|
\\
\leq &
| 1 - s M' | \| x - y \| + s L_\phi \| x - y \|
\\
= &
( | 1 - s M' | + s L_\phi ) \| x - y \| \,.
\end{align*}
Choosing $s$ such that $ s M' \leq 1$, we have that:
\[
| 1 - s M' | + s L_\phi = 1 - s M' + s L_\phi = 1 + s (L_\phi - M') < 1
\]
Hence, $F_{M',c}$ is a contractive mapping and has a unique fixed 
point $x^{M'} \in \setint{\Delta}$.
Furthermore, every function $\mathcal{M}$ from corollary 
\ref{coroll:highMutationContFunc} satisfies then $\mathcal{M}(M') = 
x^{M'}$ and thus all such functions are identical yielding the 
uniqueness of mutation equilibria for all $M > \underline{M}$.

\end{proof}

\section{Proof of proposition 
\ref{prop:mutationlimitsingames}}\label{sec:proof:mutationlimitsingames}

In order to prove proposition \ref{prop:mutationlimitsingames}, we need 
to extend our \eqref{eqn:RMD} slightly, such that we can allow more 
general mutation to occur.
Recall that $g_{ih}(x) = f_{ih}(x) - \bar{f}_{i}(x)$ and that then 
$\nashEq/ = \left\{ x \in \Delta 
\,|\, g(x) \leq 0 \right\}$ is the 
set of Nash equilibria, where the inequality is component-wise.

Then let $H = \contdiff{1}(\Delta, \RR_{>0}^S)$, and define for $c \in 
H$, $M > 
0$:
\[
[F_{M,c}(x)]_{ih} = x_{ih} + s \left( x_{ih} g_{ih}(x) + M \left( c_{ih}(x) - 
x_{ih} 
\sum_{k \leq n_i} c_{ik}(x) \right) \right)
\]
where $i \in I$, $h \in S_i$.
Note that for all $s > 0$, the fixed points of $F_{M,c}$ are the 
stationary points of a suitably generalized \eqref{eqn:RMD}.
In particular, if $c \in H$ is constant on $\Delta$, then the fixed 
points are exactly the mutation equilibria of \eqref{eqn:RMD} for a 
suitably chosen $\tilde{M}$.
It is clear that for a choice of $c \in H$, we can choose $s > 0$ such 
that for all $M \in (0,\epsilon_s)$, we have $F_{M,c}(\Delta) \subset 
\Delta$ and thus the set of fixed points is non-empty.
Therefore, we assume a suitable choice of $s>0$ (possibly depending on 
$c$).
For convenience, let us denote by $\fixedP{F_{M,c}}$ the set of fixed 
points of $F_{M,c}$ for $c \in H$ and $M > 0$, i.e.,
\[
\fixedP{F_{M,c}} = \{ x \in \Delta \,|\, F_{M,c}(x) = x \}  .
\]

From the definition of a mutation limit, we extract the main property 
and say that a set $X \subset \Delta$ has the property $(A)$, if
\begin{itemize}
\item[$(A)$] for all $c \in \setint{\Delta}$, there is a sequence of 
mutation equilibria $(x^M)_{M > 0} \subset \Delta$ that converges to an 
element of $X$.
\end{itemize}

We extend this notion to $F_{M,c}$ and say that a set $X \subset 
\Delta$ has the property $(A')$, if
\begin{itemize}
\item[{\propMut/}] for all $c \in H$ and open $U \supset X$, 
there is $M > 0$ such that $\fixedP{F_{M,c}} \cap U \neq \emptyset$.
\end{itemize}

\begin{remark*}
If is clear the a set $X$ has the property \propMut/ if and only if for 
every $c \in H$ there is a sequence $(x^M)_{M>0} \subset \Delta$ such 
that $(x^M)_{M>0}$ converges to an element of $X$ and every $x^M$ in 
the sequence satisfies $x^M \in \fixedP{F_{M,c}}$.
With this it is also clear that a set has the property $(A)$ if it has 
property \propMut/, due to $c \in \setint{\Delta}$ being equivalent to 
a constant function in $H$.
\end{remark*}

The proof of proposition \ref{prop:mutationlimitsingames} will proceed 
as follows:
We first show that $\nashEq/$ has the property \propMut/.
Next, we show that a set with the property \propMut/ contains a minimal 
set with that property, and that an analog but slightly modified result 
holds for the property $(A)$.
We then show that a minimal set with the property \propMut/ is 
connected, based on a proof by Kinoshita 
\cite{kinoshita_essential_1952}.
Thus, we have that $\nashEq/$ contains a minimal set with the 
property \propMut/, which must be contained in a connected component of 
$\nashEq/$.
Finally this set is connected and in particular has the property $(A)$ 
and hence contains a minimal connected set with the property $(A)$, 
proving proposition \ref{prop:mutationlimitsingames}.

\paragraph{Existence.}

We show first that any minimal set with the property \propMut/, must be 
contained in $\nashEq/$:

\begin{lemma}\label{lemma:NEhasA'}
Let $X \subset \Delta$ be minimal with the property \propMut/.
Then $X \subset \nashEq/$.
In particular, $\nashEq/$ has the property \propMut/. 
\end{lemma}
\begin{proof}
Assume that $X \not{\subset} \nashEq/$.
Let $c \in H$ and $(M_n)_{n \in \NN} \subset \RR_{>0}$ be a null 
sequence, and $(x^{M_n})_{n \in \NN} \subset \Delta$ convergent with 
limit $x^*$ with $x^{M_n} \in \fixedP{F_{M_n, c}}$ for all $n \in \NN$.
From our earlier note on the possibility of a constant choice of $s>0$ 
for all $n \in \NN$, and from the continuity of $g$ and $c$, we have 
that for all $i \in I$, $h \in S_i$, $x^*_{ih} g_{ih}(x^*) = 0$ holds. 

We now show that $x^* \in \nashEq/$:
If $x^* \in \setint{\Delta}$, then for all $i \in I$, $h \in S_i$, 
$x^*_{ih} g_{ih}(x^*) = 0$ implies 
$g_{ih}(x^*) = 0$, i.e., $x^* \in \nashEq/$.
If $x^* \in \partial \Delta$,
then let some $(i,h) \in S$ be such that $x^*_{ih} = 0$, and let 
$\tilde{c}_{i}
= \sup \left\{ \sum_{k \leq n_i} c_{ik}(x) \,|\, x \in \Delta \right\}$.
Then $\tilde{c}_{i} < \infty$ and for $M > 0$:
\begin{align*}
x^M_{ih} = [F_{M,c}(x^M)]_{ih}
{} = {} &
x^M_{ih} + s \left( x^M_{ih} g_{ih}(x^M) + M \left( c_{ih}(x^M) - x^M_{ih} 
\sum_{k \leq n_i} c_{ik}(x^M) \right) \right) 
\\
{} > {} &
x^M_{ih} + s \left( x^M_{ih} g_{ih}(x^M) - M x^M_{ih} 
\sum_{k \leq n_i} c_{ik}(x^M) \right)
\\
{} \geq {} &
x^M_{ih} + s x^M_{ih} \left( g_{ih}(x^M) - M \tilde{c}_{i} \right)
\\
\end{align*}

Therefore, we have for all $M>0$:
\begin{align*}
0
{} > {} &
s x^M_{ih} \left( g_{ih}(x^M) - M \tilde{c}_{i} \right)
&& \Leftrightarrow &
0
{} > {} &
g_{ih}(x^M) - M \tilde{c}_{i}
&& \Leftrightarrow &
M \tilde{c}_{i}
{} > {} &
g_{ih}(x^M)
\end{align*}
Therefore, with $M \rightarrow 0$, we have $g_{ih} (x^*) \leq 0$, and 
overall $x^* \in \nashEq/$.
Thus $X \cap \nashEq/$ has the property \propMut/ and so $X$ is not 
minimal, a contradiction.
From the fact that $x^* \in \nashEq/$, it is clear that $\nashEq/$ has 
the property \propMut/.
\end{proof}

\paragraph{Minimality.}

We first show that the existence of a set with the property \propMut/ 
implies the existence of a minimal such set, where the proof is fairly 
standard and adapted from \cite[Thm 7.3]{mclennan_advanced_2018}:

\begin{lemma}\label{lemma:minimalA'}
Let a compact set $X \subset \Delta$ have the property \propMut/.
Then it contains a minimal compact set with the property \propMut/.
\end{lemma}
\begin{proof} 
The proof is based on Zorn's lemma.
Let $C$ be the set of compact subsets of $X$ with the property 
\propMut/, 
i.e.,
\[
 C = \left\{ K \subset X \,|\, K \neq \emptyset \text{ and $K$ is 
 compact and has the property \propMut/} \right\},
\]
and order $C$ by reverse inclusion $\supset$.
Let $O \subset C$ be completely ordered.
Then $O$ has the finite intersection property, as it is completely 
ordered by reverse inclusion and its elements are compact.
Therefore, $K_\infty := \bigcap O \neq \emptyset$ and $K_\infty$ is 
compact.

It remains to show that $K_\infty$ has the property \propMut/:
Assume $K_\infty$ does not have the property \propMut/.
Then there is a $c \in H$ and an open neighbourhood $V$ of 
$K_\infty$ such that no $F_{M,c}$ ($M > 0$) has a fixed point in $V$.
For $L \in O$, we have $L \not\subset V$ because $L$ has the property 
\propMut/.
Then $O' := \{ L \backslash V : L \in O \}$ is a completely ordered 
collection of compact sets ($L$ is compact and $V$ is open) with
the finite intersection property, inherited from the reverse inclusion 
ordering of $O$.
Therefore, it has a nonempty intersection $K'_\infty \subset K_\infty 
\subset V$ but $K'_\infty \cap V = \emptyset$, which is a contradiction.
Thus, $K_\infty$ has the property \propMut/ and therefore $K_\infty \in 
C$ is an upper bound of $O$.

With Zorn's lemma then, $C$ has a maximal element, which is a minimal 
compact subset of $X$ with the property \propMut/.
\end{proof}

For the existence of a mutation limit we will have to make a 
similar step, however preserving connectedness:

\begin{lemma}\label{lemma:minimalconnectedA}
Let a connected compact set $X \subset \Delta$ have the property $(A)$.
Then it contains a minimal connected compact set with the property 
$(A)$, i.e., a mutation limit.
\end{lemma}
\begin{proof}
Let $C$ be the set of all compact connected (non-empty) subsets of $X$ 
with the property $(A)$, partially ordered by $\supset$ and $O$ a 
completely ordered chain in $C$.
Then $K_\infty = \bigcap_{K \in O} K$ is non-empty, compact and has the 
property $(A)$ by an argument completely analogous to the previous 
lemma.

It remains to show that $K_\infty$ is connected:
Assume that $K_\infty$ is not connected.
Then, there are open disjoint sets $U_1$, $U_2$, with $K_\infty \subset 
U_1 \cup U_2 =: U$, with $U$ open in $X$.
$X$ and all $K \in O$ are compact and, with $X$ being Hausdorff, also 
closed.
Thus $X \setminus K$ is open in $X$ for $K \in O$.
Then, with
$
\bigcup_{K \in O} X \setminus K = X \setminus \bigcap_{K \in O} K = X 
\setminus K_\infty,
$
we have that $\{ U \} \cup \{ X \setminus K | K \in O \}$ is an open 
cover of $X$,
and there is a finite subcover $\{ U \} \cup \{ X \setminus K_i | K_i 
\in O, 1 \leq i \leq n \}$, as $X$ is compact.
Thus $X = U \cup \bigcup_{1 \leq i 
\leq n} X \setminus K_i = U \cup X \setminus \bigcap_{1 \leq i 
\leq n} K_i $. As $O$ is completely ordered by inclusion, we can assume 
that $K_1 \supset K_i$ ($1\leq i \leq n$) and we have that $X = 
U \cup X \setminus K_1 $.
Thus $K_1 \subset U = U_1 \cup U_2$, and hence $K_1$ is not connected, 
a contradiction.
Therefore, $K_\infty$ is connected and $K_\infty \in C$.
With Zorn's lemma the statement of the lemma follows.
\end{proof}

\paragraph{Connectedness.}
We gain connectedness as a necessary property of minimal sets with 
the property \propMut/, where the main idea of the proof is based on a 
proof by Kinoshita \cite{kinoshita_essential_1952} and relies on the 
``convexity'' of $H$:

\begin{lemma}\label{lemma:minA'connected}
If $K$ has the property \propMut/ and $K = (K_1 \cup \ldots \cup K_s)$ 
with the $K_j$ disjoint and compact, then some $K_j$ has the property 
\propMut/.
\end{lemma}
\begin{proof}
Let $K$ have property \propMut/ and $K = K_1 \cup \ldots \cup K_s$ with 
the $K_j$ disjoint and compact. 
Assume that no $K_j$ has the property \propMut/.
Then there are $c_1, \ldots, c_s \in H$ and neighbourhoods
$U_1, \ldots, U_s$ of $K_1, \ldots, K_s$ with disjoint closures
such that for all $M > 0$, $\mathscr{F}(F_{M,c_j}) \cap U_j = 
\emptyset$. 
Let further $V_1, \ldots, V_s$ be strictly smaller neighbourhoods, 
i.e., $ \overline{V}_j \subsetneq U_j $,
and let $U_0$ be a neighbourhood of
$\Delta \backslash (U_1 \cup \ldots \cup U_s)$ whose closure is 
disjoint from the $V_1, \ldots, V_s$, and $c_0$ any function in $H$.
Then $\{U_0, U_1, \ldots, U_s\}$ is an open cover of $\Delta$ and 
with $\Delta$ being a compact subset of a topological vector space, 
there is a $C^\infty$-partition of unity
$\pi_0, \pi_1, \ldots, \pi_s$ such that
$ \pi_j(x) = 0$ ($\forall x \in \Delta \setminus U_j$),
and $\sum_{j = 0}^{s} \pi_j(x) = 1 $ (${\forall x \in \Delta}$),
e.g., \cite[Thm 6.2]{mclennan_advanced_2018}.
The convex combination, $\bar{c}$, with
$\bar{c}: x \mapsto \sum_{j = 0}^{s} \pi_j(x) c_j(x)$, 
is an element of $H$.
Considering $F_{M, \bar{c}}$, we then have that $F_{M, \bar{c}}(x) = 
F_{M, c_j}(x)$ for $x \in V_j$.
Thus $\fixedP{F_{M, \bar{c}}} \cap V_j = \emptyset $ for
$1 \leq j \leq s$ for all $M > 0$.
Therefore, $F_{M, \bar{c}}$ has no fixed points in
$(V_1 \cup \ldots \cup V_s) \supset K$ for any $M > 0$.
This is a contradiction to the assumption that $K$ has the property 
\propMut/.
\end{proof}

This result gives us the connectedness of minimal sets with the 
property \propMut/:

\begin{coroll}\label{coroll:minA'connected}
Let $X$ be minimal with the property \propMut/. Then $X$ is connected.
\end{coroll}

Overall, this proves the following:
\begin{prop}
There is a mutation limit.
In particular, it is contained in a connected component of 
$\nashEq/$.
\end{prop}
\begin{proof}
With lemma \ref{lemma:NEhasA'}, $\nashEq/$ has the property 
\propMut/.
With $\nashEq/$ being compact due to $g \in C(\Delta, \RR^S)$ and 
$\nashEq/ \subset \Delta$, and with lemma \ref{lemma:minimalA'}, 
there is a minimal compact set $X' \subset \nashEq/$ with the 
property \propMut/.
Furthermore, with corollary \ref{coroll:minA'connected}, $X'$ is 
connected.
With the property \propMut/, $X'$ also has the property $(A)$.
With lemma \ref{lemma:minimalconnectedA}, $X'$ contains a minimal 
connected compact subset $X \subset X'$ with the property $(A)$.
By definition, $X$ is a mutation limit.
\end{proof}

\end{document}